\documentclass[11pt]{article}

\usepackage{amsmath,amsthm,amssymb,bbm,mathtools}
\usepackage{authblk}
\usepackage{verbatim}
\usepackage{thmtools, thm-restate}
\usepackage{tikz}
\usepackage{stmaryrd}
\usepackage{subcaption}
\usepackage{bbm}
\usepackage[left=2.4cm,right=2.4cm,top=2.4cm,bottom=2.4cm]{geometry}
\usepackage[]{color}
\usepackage{algorithm,algorithmicx,algpseudocode}
\usepackage{etoolbox}

\newcommand{\black}{\black}

\usepackage[font=small,labelfont=bf]{caption}
\usepackage[bookmarks, colorlinks, breaklinks,
pdftitle={},pdfauthor={}]{hyperref}

\hypersetup{
  linkcolor=black,
  citecolor=black,
  filecolor=black,
  urlcolor=black}

\usepackage{cleveref}

\declaretheorem[]{theorem}
\declaretheorem[sibling=theorem]{lemma}
\declaretheorem[sibling=theorem]{corollary}

\declaretheorem[sibling=theorem]{proposition}
\declaretheorem[numbered=yes]{remark}
\declaretheorem[sibling=theorem]{definition}

\numberwithin{equation}{section}

\usetikzlibrary{calc,decorations.markings}
\usetikzlibrary{decorations.pathmorphing}
\usetikzlibrary{decorations.pathreplacing}
\usetikzlibrary{patterns}
\usetikzlibrary{arrows}

\newcommand{\bb}[1]{\mathbb{#1}}

\newcommand{\blank}[1]{}

\newcommand{\R}{\bb R}
\newcommand{\Z}{\bb Z}

\newcommand{\N}{\bb N}
\renewcommand{\P}{\bb P}

\newcommand{\abb}[1]{\left| {#1} \right|}

\newcommand{\bc}[1]{\left\{ {#1} \right\}}
\newcommand{\pa}[1]{\left( {#1} \right)}

\newcommand{\al}[0]{\alpha}
\newcommand{\be}[0]{\beta}
\newcommand{\ga}[0]{\gamma}

\newcommand{\de}[0]{\delta}
\newcommand{\De}[0]{\Delta}
\newcommand{\ep}[0]{\epsilon}

\newcommand{\La}[0]{\Lambda}

\newcommand{\Te}[0]{\Theta}

\newcommand{\Om}[0]{\Omega}
\newcommand{\si}[0]{\sigma}

\newcommand{\rc}[1]{\frac{1}{#1}}
\newcommand{\prc}[1]{\pa{\rc{#1}}}

\newcommand{\fc}[2]{\frac{#1}{#2}}
\newcommand{\sfc}[2]{\sqrt{\frac{#1}{#2}}}
\newcommand{\pf}[2]{\pa{\frac{#1}{#2}}}

\newcommand{\sub}[0]{\subset}
\newcommand{\bs}[0]{\backslash}
\newcommand{\vocab}[1]{\textbf{#1}}

\newcommand{\prodo}[2]{\prod_{#1=1}^{#2}}

\newcommand{\one}[0]{\mathbbm{1}}
\newcommand{\pl}[0]{\partial}
\newcommand{\iy}[0]{\infty}

\newcommand{\subeq}[0]{\subseteq}

\newcommand{\Var}[0]{\operatorname{Var}}

\newcommand{\lra}[0]{\leftrightarrow}

\newcommand{\ol}[1]{\overline{#1}}

\newcommand{\ce}[1]{\left\lceil {#1}\right\rceil}

\newtoggle{appendix}
\toggletrue{appendix}

\title{Approximation algorithms for the random-field Ising model}
\author{Tyler Helmuth\thanks{Durham University, Mathematical Sciences Department}, Holden Lee\thanks{Duke University, Department of Mathematics}, Will Perkins\thanks{University of Illinois at Chicago, Department of Mathematics, Statistics, and Computer Science. Supported in part by NSF grants DMS-1847451 and CCF-1934915.}, Mohan Ravichandran\thanks{Bogazici University, Department of Mathematics}, Qiang Wu\thanks{University of Illinois at Urbana-Champaign, Department of Mathematics.}}
\date{\today}

\begin{document}
\maketitle

\begin{abstract}
  Approximating the partition function of the ferromagnetic Ising
  model with general external fields is known to be \#BIS-hard in the
  worst case, even for bounded-degree graphs, and it is widely
  believed that no polynomial-time approximation scheme exists.  This
  motivates an average-case question: are there classes of instances
  for which polynomial-time
  approximation schemes exist?  We investigate this question for the
  random field Ising model on graphs with maximum degree
  $\Delta$.  We establish the existence of fully polynomial-time
  approximation schemes and samplers with high probability over the
  random fields 
  if the external fields  are IID Gaussians with variance
  larger than a constant depending only on the inverse temperature and
  $\Delta$.  
  The main challenge comes from the positive density
  of vertices at which the external field is small. These regions,
  which may have connected components of size 
  $\Theta(\log n)$, are a barrier to algorithms based on establishing a
  zero-free region, and cause worst-case analyses of Glauber dynamics
  to fail.   
  The analysis of our algorithm is based on
  percolation on a self-avoiding walk tree. 
\end{abstract}

\section{Introduction}
\label{sec:introduction}

Recent years have seen the development of a rich interplay between statistical physics,
computational complexity, and algorithm design. One central question 
is the extent to which phase transitions for discrete statistical
mechanics models are related to the tractability of associated
computational problems. In this paper we are primarily interested in
approximate counting and sampling;  
see Section~\ref{sec:defn} for formal definitions. 
Results concerning these tasks have
traditionally focused on (i) establishing algorithmic tractability in
so-called `high temperature' (weakly correlated) regimes, and
(ii) establishing the failure of certain algorithmic techniques in
`low temperature' (strongly correlated) regimes. Very recently,
positive algorithmic results have been
obtained at low temperatures~\cite{HelmuthPerkinsRegts,BarvinokRegts,huijben2021sampling},
and occasionally even at all
temperatures~\cite{jerrum1993polynomial,BCHPT,HJP,ALOVII}.

Many of these algorithmic results have 
been shown using 
a detailed probabilistic and physical understanding of the corresponding
statistical mechanics problems. The intuition gained from this
understanding is typically restricted to specific classes of
graphs, e.g., lattices or specific models of random graphs. One of the
challenges for algorithm design is to go beyond these restricted
classes of graphs, and this can lead to situations in which the
notions of `high temperature', `low temperature', and `phase
transition' are unclear. 

This loss of intuition is present 
when one asks about the \emph{average-case complexity} of a
well-known \#BIS-hard problem, 
the ferromagnetic Ising model
with general vertex-dependent external fields.  Towards understanding this situation, in this paper 
we  consider the  Ising model with
vertex-dependent \emph{random} external fields $h_{x}$, $x \in V(G)$,
where the $h_{x}$ are independent and identically distributed (IID) 
centered Gaussian.
This model, primarily studied in statistical physics on the integer
lattice $\mathbb Z^d$, is known as the \emph{random field Ising
  model}. A great deal of interest in the random field Ising model has arisen because it behaves differently than the zero-field model with $h_x\equiv0$. To briefly describe this, recall that the zero-field model undergoes a phase transition on $\Z^d$ when $d\geq 2$: if $\beta>0$ is small, then correlations decay exponentially and there is a unique infinite volume Gibbs measure. On the other hand, if $\beta$ is large, then correlations do not decay, and multiple Gibbs measures exist. The surprising phenomena is that this picture changes for the random field Ising model: on $\Z^d$ for $d\geq 3$ there is still a phase transition if the variance of the random fields $h_x$ is not too large~\cite{BricmontKupiainen}, but on $\Z^2$ there is no phase transition if the variance is non-zero~\cite{AizenmanWehr}. In fact, in  recent breakthroughs, it was shown that on $\Z^2$ correlations always decay exponentially~\cite{DingXia}, and on $\Z^d$, $d \ge 3$, correlations decay exponentially throughout the high-temperature regime~\cite{ding2021new}. It is natural to wonder if there are algorithmic counterparts to these physical phenomena.

While the phase transition phenomena of the previous paragraph concerned small variances, the random field Ising model also exhibits interesting properties (so-called Griffiths singularities) from the point of view of physics when the variance of the external fields is not small.
See Section~\ref{sec:background}. This regime is also \emph{terra incognita} from an algorithmic point of view, and we focus in this paper on the large variances. 

It is straightforward, see Section~\ref{sec:largefield}, to design
algorithms for the random field Ising model if one
assumes $|h_{x}|$ is \emph{uniformly} large (depending on the inverse
temperature $\beta$ and on the maximum degree of the graphs being
considered): these large external fields put the system in an
effectively high temperature situation. If, however, $|h_{x}|$ can be
small for some vertices $x$, then 
for large inverse temperature $\beta$, highly correlated
subsets of spins may appear --- there can be large `low temperature' islands in a `high
temperature' sea. Our main result shows that if $|h_{x}|$ is
  \emph{typically} large 
then for typical realizations of the external fields, the computational
tasks of approximate counting and sampling are tractable. We will
discuss our proof strategy and highlight the barriers faced by other
methods after pausing to give precise formulations of our results.

\begin{definition}
Let $G=(V,E)$ be a finite graph, $h\colon V\to \R$, and
$\be\in \R$. The \vocab{Ising model} on $G$ with inverse temperature $\beta$ and external fields $h$ is the probability distribution on $\{\pm 1 \}^V$ given by
\begin{equation}
    \label{e:Idef}
    p_{G,\beta,h}(\sigma) = \frac{e^{-H_{\beta,h}(\sigma)}}{Z_{G,\beta,h}}, \quad
    Z_{G,\beta,h} = \sum_{\sigma\in \{\pm 1\}^V} e^{-H_{\beta,h}(\sigma)},
\end{equation}
where the \vocab{Hamiltonian} $H_{\beta,h}$ is the function
\begin{equation*}
    -H_{\beta,h}(\sigma) = \beta\sum_{xy\in E}\sigma_x \sigma_y + \sum_{x\in V} h_x\sigma_x.
\end{equation*}
\end{definition}
A \vocab{random field Ising model} has external fields $h$ that are independent random variables with  prescribed distributions. We use $\P$ to denote the law of the random external fields.  In this paper we will be concerned with random
external fields that are  typically large, which we will model by  centered Gaussians with large variance. Our main theorem in this setting 
is the following. 
\begin{theorem}
  \label{thm:main}
  For every $\Delta \ge 2$, $\beta \in \R$, there exists $H = H(\De,\be)$
  large enough so that for random field Ising models with inverse
  temperature $\beta$ and external 
  fields distributed as independent
  $\mathcal{N}(0,H)$ random variables the following holds.  For every graph $G $ of
  maximum degree $\Delta$ on $n$ vertices, 
  with probability $1-o(1)$ over the choice of
  random fields, there exists an FPTAS for $Z_{G,\beta, h}$ and
  a polynomial-time sampling scheme.
\end{theorem}

\begin{remark}
  The failure probability (over the randomness of the fields) tends to
  $0$ with the size of the graph $n$.  We can make this failure
  probability arbitrarily small: to achieve failure probability
  $\delta$ requires a factor polynomial in $1/\delta$ in the running
  time of the algorithm.
\end{remark}

\begin{remark}
  The external fields being Gaussian does not play a role in our proof. Theorem~\ref{thm:main} applies more generally to independent 
  external fields with 
  distributions with the property that 
\begin{equation*}
  \label{e:hyp2}
  \P(|h_x|< c)\le p. 
\end{equation*}
for 
  $c =|\be| \De + \log \De + c_1$ and
  $p = \fc{c_2}{\De^2}$
  for large enough constant $c_1$ and small enough constant $c_2$. 
In particular, since $\P_{h\sim N(0,\si^2)} (|h|\le c) \le \sfc{2}{\pi} \fc{c}{\si}$, we see that $H=\Om(\be^2\De^6)$ 
suffices in the case of Gaussian external fields and $\be$ bounded away from 0. 
\end{remark}

\begin{remark}
  \label{rem:check}
  We can efficiently check whether a given instance of the 
  external fields satisfy the required conditions in the following
  sense: given $\epsilon>0$ it takes time polynomial in $n$ and
  $1/\epsilon$ to both output an approximation of the partition
  function and to check the conditions that guarantee the
  $\epsilon$-relative accuracy of the approximation.  See
  Proposition~\ref{prop:check}. We emphasize however, that with probability
  $1-o(1)$, a single instance satisfies these conditions for all
  choices of $\epsilon$.
\end{remark}

\begin{remark}
  Theorem~\ref{thm:main} extends in a straightforward manner to the
  setting of edge-specific inverse temperatures $\be_{xy}$ provided
  all $\beta_{xy}$ are bounded in absolute value by a fixed
  $\be>0$. We consider a single inverse temperature for notational
  simplicity.
\end{remark}

\begin{remark}
  Theorem~\ref{thm:main} also applies in the presence of boundary conditions, which arise naturally in our analysis. We define boundary conditions formally in Section~\ref{sec:defn}. 
\end{remark}

The key mechanism behind the proof of Theorem~\ref{thm:main} is that
large external fields cause the system to rapidly decorrelate, as
spins tend to align with their external field. We formalize this
decorrelation by generalizing a disagreement percolation argument due
to Camia, Jiang, and Newman~\cite{camia2018note}. The resulting notion
of correlation decay is similar to, but somewhat weaker than, strong
spatial mixing. See Section~\ref{sec:ssmtree}. This correlation decay
property is sufficiently strong to enable a recursive analysis on a
self-avoiding walk (SAW) tree as was pioneered by
Weitz~\cite{weitz2006counting}. Weitz's method for approximately counting weighted independent sets in bounded degree graphs is now known as the `method of correlation decay', and has found numerous applications in approximate counting and sampling and in proving Gibbs uniqueness for spin models on $\mathbb Z^d$, e.g,~\cite{bayati2007simple,sinclair2014approximation,li2013correlation,restrepo2013improved,sinclair2017spatial}.  The analysis of correlation decay algorithms involves proving strong spatial mixing on the SAW tree, usually by means of a contraction argument or a monotonicity argument with respect to boundary conditions.  What is new in our approach is proving a form of strong spatial mixing on the appropriate SAW tree by a probabilistic argument based on percolation theory (more precisely, based on disagreement percolation~\cite{van1993uniqueness,van1994percolation}).

The proof of Theorem~\ref{thm:main} is fairly robust and can be generalized to apply to graphs where the maximum degree is not necessarily bounded. To illustrate this, we use similar ideas to
analyze the random field Ising model on sparse Erd\H{o}s-R\'enyi random graphs. Recall that a graph drawn from $\mathcal G(n,p)$ (an \vocab{Erd\H{o}s-R\'enyi random graph}) is defined as a graph on $n$ vertices $\{v_1,\ldots, v_n\}$ where each potential edge 
$\{v_i,v_j\}$, $i\ne j$ independently included 
with probability $p$.

\begin{theorem}\label{thm:main2}
For every $\De >1$ and $\beta \in \R$, there exists $H$ large enough so that the following holds.  With $G\sim \mathcal G(n,\De/n)$ and independent random external fields distributed as $\mathcal{N}(0,H)$, with probability $1-o(1)$ over the random graph and random fields there exists an FPTAS and polynomial-time sampling scheme for the random field Ising model on $G$ at inverse temperature $\beta$. 
\end{theorem}

\begin{remark}
Theorem~\ref{thm:main2} applies more generally to independent external fields with distributions with the property that 
\begin{equation*}
  \label{e:hyp}
  \P(|h_x|< c)\le p. 
\end{equation*}
for 
  $c = c_1(|\be| \De + \log \De + 1)$ and
  $p = \prc{2\De}^{c_2}$
  for constants $c_1,c_2$. 
\end{remark}

\subsection{Background}
\label{sec:background}

An important challenge for understanding the relative complexity of approximate
counting was raised by Dyer, Goldberg, Greenhill, and Jerrum in~\cite{dyer2004relative}: how difficult is it to
approximately count independent sets in \emph{bipartite} graphs? This
problem, \vocab{(approximate) \#BIS}, occupies a central place in the analysis
of approximate counting algorithms~\cite{GoldbergJerrum,galanis2016ferromagnetic,liu2014complexity}. 
The most relevant fact for us is the \#BIS-hardness of
approximately computing the partition function of the ferromagnetic
Ising model with general vertex-dependent external
fields on bounded-degree bipartite graphs~\cite{cai2016hardness}.

For \#BIS-hardness, 
allowing general external fields is
necessary: 
it is a classic result 
that there are efficient
approximation algorithms for the ferromagnetic Ising model with no
external fields or with consistent external fields (all non-negative or all non-positive) for all values of
$\beta\geq 0$ of the inverse temperature~\cite{jerrum1993polynomial,GoldbergJerrum}. The novelty of
Theorem~\ref{thm:main} is
that it allows for inconsistent
external fields. Standard approaches for analyzing the Glauber
dynamics of the Ising model do not seem capable of proving
Theorem~\ref{thm:main} --- see Remark~\ref{rem:Glauber} below. Note that
if $\be$ is taken sufficiently small, then standard high
temperature methods already apply~\cite{zhang2011approximating}, and
thus the most interesting case of our theorem is when $\be$ is
large.

Another approach to approximation is based on zero-freeness of the
partition function, either via Barvinok's method~\cite{Barvinok} or
cluster expansion methods~\cite{HelmuthPerkinsRegts}. 
The main barrier to applying these methods in the case of the random field Ising model is the phenomena of Griffiths 
singularities~\cite{vanEnter}. 
These singularities arise in spin
systems with random Hamiltonians; 
in our case the randomness is contained in the external field.
When the underlying graph is the integer lattice $\mathbb Z^d$, the existence of rare (but
arbitrarily large) regions of atypical behaviour for the random field is
widely believed to lead to thermodynamic functions being infinitely
differentiable but \emph{not}
analytic~\cite{von1995taming,frohlich1984improved,vanEnter}. The
non-analyticity of limiting quantities rules out the existence of
zero-free regions in finite volumes.

\subsection{Future Directions}
\label{sec:disc-future-direct}
Our algorithm is based on Weitz's method of correlation decay on a
computational tree. A natural question is whether Markov chain-based
algorithms 
can provide a similar guarantee. See Remark~\ref{rem:Glauber} for indications this may be a subtle question. Our thresholds are certainly improvable, and the
tractability for more moderate values of external field is
unknown, as is tractability in the presence of correlated external fields. 

The difference in behaviour for the RFIM in $d=2$ and
$d\geq 3$ suggests that the design of approximate counting algorithms
in the presence of weak disorder is a subtle task, and hence an interesting challenge for future research.
Another interesting direction is to develop algorithms for problems that contain
`high temperature' islands in a `low temperature' sea, i.e., with the
roles of high and low temperature in the present paper exchanged. 
Our result does not rely on ferromagnetism, and it is a good
question whether one can obtain a stronger result---such as one that
works for more moderate external fields --- in the ferromagnetic $\beta> 0$ regime. 

We end this section by indicating a motivating connection (and potential future direction) between the results of this paper and
\#BIS that does not pass through any formal reductions as in~\cite{cai2016hardness,GoldbergJerrum}. A
difficulty in investigating the complexity of \#BIS is that it is unclear 
which instances are hard.  For a single random bipartite graph
(balanced or not), the low temperature behaviour of independent sets
is well understood (see, e.g.,~\cite{mossel2009hardness}):  independent sets typically consist of significantly more 
vertices on one side of the bipartition than the other. Thus, in
the search for a hard instance one may be tempted to treat single random
bipartite graphs as gadgets, and to assemble many gadgets together by
adding edges between the gadgets in a bipartite manner. If the density
of added edges is low enough to avoid disrupting the behaviour of
individual gadgets, then in the low temperature regime
the resulting graph heuristically behaves like a ferromagnetic Ising
model. The external field reflects if the constituent graphs are
balanced ($h=0$) or not ($h\neq 0$). Working directly with the Ising
model with an inconsistent magnetic field allows for us to search for
hard instances while bypassing the technicalities that would be
present in making the preceding discussion precise.

\subsection{Organization of the paper}
In Section~\ref{sec:largefield}, we give approximate sampling and counting algorithms in the case that all external fields are large with probability one.
In Section~\ref{sec:ssmtree}, we prove our main theorem (Theorem~\ref{thm:main}). In Section~\ref{sec:extensions}, we prove Theorem~\ref{thm:main2}, the extension of our main result to random graphs. Finally, in Section~\ref{s:non-unif}, we show that our work generalizes one of the main theorems of~\cite{camia2018note} to infinite graphs of max-degree $\Delta$.
In Appendix~\ref{a:saw}, we give details of the SAW tree construction and recursion used by the algorithms, and in Appendix~\ref{a:alg}, we write out the algorithms explicitly.

\subsection{Preliminaries and notation}
\label{sec:defn}

\paragraph{Approximate counting and sampling.} 

A \vocab{fully polynomial-time approximation scheme (FPTAS)} for a function $Z(G)$ is a deterministic algorithm that given a graph $G$ and a tolerance $\epsilon>0$ outputs a number $\hat Z$ such that $e^{-\epsilon}\hat Z \leq Z\leq e^{\epsilon}Z$, with running time polynomial in $1/\epsilon$ and $|V(G)|$. A \vocab{polynomial time sampling scheme} for a distribution $\mu_G$ is a randomized algorithm that, given $G$ and a tolerance $\epsilon>0$ outputs a sample from a distribution $\hat\mu$ such that $\|\hat\mu - \mu_G\|_{\textrm{TV}}\leq \epsilon$, with running time polynomial in $1/\epsilon$ and $|V(G)|$. 

\paragraph{Notation.} Throughout we implicitly restrict our attention to
  connected graphs, as all of the algorithmic tasks we consider factor
  over connected components. We let $\mathcal G_\De$ denote the set of graphs with maximum degree at most $\De$, and write $\deg(v)$ for the degree of a vertex $v$. Given a graph $G=(V,E)$, we denote by $d(v,w)$ the distance between vertices $v,w$ on the graph, i.e., the length $\ell$ of the shortest  path $v_0,\ldots, v_\ell$ with $v_0=v$ and $v_\ell=w$, and $(v_i,v_{i+1}) \in E$.
For a set $S\sub V$, let $d(v,S):= \min_{w\in S} d(v,w)$. For a vertex $v$ let $N(v)$ denote the set of neighbors of $v$. Let $N(v,\ell)$ denote the set of vertices at distance exactly $\ell$ in  a graph $G$. 

The \vocab{Ising model with boundary condition $\tau\in \{\pm 1\}^V$ on $B\subset V$} is defined by the formulas in~\eqref{e:Idef} but with the restriction that $\sigma$ is a spin configuration that agrees with $\tau$ on $B$. We write $p^{\tau}_{G,\be,h}$ for the law of this model.
Given an Ising model with fixed $\tau, \beta, h$, and letting $\si'\in \{\pm 1\}^\La$ for a subset $\La\sub V$, let $p_v^{\si'}$ be the marginal probability of spin 1 at vertex $v$ conditioned on $\si_\La = \si'$, i.e.,
\begin{equation}
\label{e:isingmarginal}
p_v^{\si'}:= p^{\tau}_{G, \be, h}(\si_v=1|\si_{\La} = \si').
\end{equation}
In~\eqref{e:isingmarginal} and above we have written $\sigma_A = (\sigma_x)_{x\in A}$ to denote the spins at the vertices $A\subset\Lambda$.

For functions $f,g \colon \R\to\R$ we write $f=O(g)$ if there exists $C>0$ such that $|f(x)|\leq C |g(x)|$ for all $x$ large enough, and $f=\Omega(g)$ if $g=O(f)$.

\section{Deterministic large external fields}
\label{sec:largefield}
In this section we give approximate sampling and counting algorithms
in the case that $|h_x| \ge c(\beta,d)$ for all $x \in V$.
 This case in which \emph{all} external fields are large with probability $1$ provides some intuition for the main result by indicating how the
presence of large fields facilitates correlation decay. However, the
simple proof we provide here 
does not work without a uniform bound on the external fields, 
see Remark~\ref{rem:Glauber} below.  Define 
\begin{equation}
  \label{eq:M}
M(\De, h, \be) =\abb{ \rc{1+e^{-2\be \De - 2h}} - \rc{1+e^{2\be \De - 2h}}}.
\end{equation}
The quantity $M$, and particularly upper bounds on $M$, will be important for our
analysis in this and subsequent sections.
\begin{lemma}\label{l:M}
For any $\be \in \R$,  $\De\ge 0$, and $\ep>0$, if $|h|\ge |\be| \De + \rc 2 \log\prc \ep$, then $M(\De, h, \be)< \ep$.
\end{lemma}
\begin{proof}
Consider the terms $\rc{1+e^{-2\be \De - 2h}}$ and $\rc{1+e^{2\be \De - 2h}}$. 
If $h\ge |\be| \De + \rc 2 \log\prc \ep$, then both terms are $\ge
\rc{1+\ep}$, and if $h\le -(|\be| \De + \rc 2 \log\prc \ep)$, then both
terms are $\le \rc{1+\rc{\ep}}= \fc{\ep}{1+\ep}$. Since both terms are
bounded between 0 and 1, $M(\De, h, \be)\le \fc{\ep}{1+\ep}<\ep$ follows.
\end{proof}

The following monotonicity property follows by differentiating in $\Delta$.
\begin{lemma}
\label{lem:Mmon}
 For fixed $h,\beta$, $M(\Delta,h,\beta)$ is non-decreasing in $\Delta$.
\end{lemma}

The following bounds the influence of boundary conditions on the marginal probability at $v$. 
\begin{lemma}\label{l:marg-M}
Let $v\in V$, $\La\sub V$ be a set not containing $v$, and
let $\si_\La, \tau_\La\in \{\pm 1\}^{\La}$.  Then
\begin{equation*}
|p^{\si_{\La}}_v - p^{\tau_{\La}}_v| \le M(\deg(v), h_v, \be). 
\end{equation*}
\end{lemma}
\begin{proof}
We first prove the result when $\La$ contains all the neighbors of $v$. 
In this case, considering only the relevant part of the Hamiltonian and temporarily abbreviating $\sigma=\sigma_{\Lambda}$ on the right-hand side,
\begin{equation*}
p_v^{\si_{\Lambda}} = \fc{e^{\sum_{y\in N(v)} \be \si_y+h_v}}{e^{\sum_{y\in N(v)} \be \si_y +h_v} + e^{\sum_{y\in N(v)} -\be \si_y-h_v}}.
\end{equation*}
The maximum and minimum possible values of this are $\rc{1+e^{-2\be\deg(v) - 2h_v}}$ and $\rc{1+e^{2\be\deg(v) - 2h_v}}$. Hence $|p_v^{\si_\La} - p_v^{\tau_\La}|$ is bounded by $M(\deg(v),h_v,\be)$. 
For general $\La$, by conditioning on the value of $\si$ on $N(v)$, we can write $p_v^{\si_\La}$ as a weighted average of $p_v^{\si_{N(v)}}$, so the lemma follows in this case as well.
\end{proof}

This bound allows us to prove rapid mixing of Glauber dynamics via the
method of path coupling~\cite{bubley1997path}.  This in turn gives us
randomized polynomial-time approximate counting and sampling
algorithms when the external field is \emph{uniformly}
large. 
We recall that the Glauber dynamics are a
time-homogeneous Markov chain $(\sigma(n))_{n\in \N}$ that evolves by uniformly selecting a vertex $x\in V$, and then updating
the spin at $x$ according to the marginal distribution at $x$ conditioned on the spins of its neighbors, i.e., $\P(\sigma_{x}(n+1)=1) =
p_{\beta,h}^{\tau}(\sigma_{x}=1|\sigma_{V\setminus
  \{x\}}=\sigma(n)_{V\setminus
  \{x\}})$. The other spins are unchanged in this step.
  
\begin{theorem}
  \label{thm:PC}
 Fix $\De\in\{2,3,\dots\}$. 
  If $|h_{x}|\geq h_{0}(\Delta,\beta)=\Delta|\be|+
  \frac{1}{2}\log\Delta$, then the mixing time of the Glauber dynamics 
  for the Ising model on $G\in\mathcal{G}_{\Delta}$ with $|V(G)|=n$ is $O( n \log n)$.  
\end{theorem}
\begin{proof}
  Fix $G\in \mathcal{G}_{\Delta}$. As described above, in a single
  step of the Glauber dynamics we pick $x \in V(G)$ uniformly at
  random and then update the spin $\sigma_x$ conditioned on the spins
  of $N(x) = \{y\in V \mid \{x,y\}\in E(G)\}$.

  Consider two configurations $\sigma, \sigma'$ that disagree only at
  vertex $y$.  We couple two copies of the Glauber dynamics starting
  from $\sigma(0) = \sigma$ and $\sigma'(0)=\sigma'$ respectively by
  picking the same vertex $x$ to update and updating to the same spin
  with as high probability as possible. We analyze how the Hamming
  distance between the two configurations changes in a single step of
  the chain.
  \begin{enumerate}
  \item If $x= y$, the vertex at which $\si$ and $\si'$ disagree, then
    with probability one $\sigma(1)= \sigma'(1)$, and the Hamming
    distance decreases by $1$. This occurs with probability $1/n$.
  \item If $d(x,y) >1$, then both chains see the same boundary conditions
  and so make the same update. The Hamming distance does not change.
  \item If $x \in N(y)$, then the two chains see different boundary
    conditions. 
    By Lemmas~\ref{lem:Mmon} and~\ref{l:marg-M} the difference in the
    probability of updating to $+1$ is at most $M(\De, h_v,\be)$,
    and this is an upper
    bound on the expected change in Hamming distance.  This occurs with
    probability $\De/n$.
  \end{enumerate}
  Let $\de(\sigma,\sigma')$ be the expected change in Hamming distance
  between $\sigma$ and $\sigma'$ after one step
  of the coupled chains. Since $|h_{v}|\geq h_{0}(\Delta,\beta)$,
  Lemma~\ref{l:M} and the considerations above give
  \begin{equation*}
    \de(\sigma,\sigma') \le - \frac{1}{n} \left[ 1 - \Delta
      M(\De,h_{0},\beta)  \right] <0. 
  \end{equation*}
  The theorem follows by path coupling, see,
  e.g.,~\cite[Corollary~14.7]{levin2017markov}. 
\end{proof}

\begin{remark}
  \label{rem:Glauber}
  Analyzing the Glauber dynamics without a uniform lower bound on
  $|h_{x}|$ would require addressing the fact that the dynamics are
  not contractive at each step. To see this, suppose that
  $h_{x}\in \{-h,0,h\}$, and note $M(\Delta,0,\be)\approx 1$ if $\be$ is
  large. In this case when the dynamics act on vertices with $h_{x}=0$ there
  is no contraction. This reflects the fact that $\be$
  is above the uniqueness threshold for the Ising model on the $\Delta$-regular tree with no
  external field. 
  
  For some types of
  disordered systems, rigorous results that rule out exponential
  relaxation have been obtained in infinite
  volume, see, e.g.,~\cite{cesi1997relaxation}.
\end{remark}

\section{Random field Ising model} 
\label{sec:ssmtree}

In this section we prove Theorem~\ref{thm:main}. We will that show strong spatial mixing, a strong form of
correlation decay, holds with high probability on a tree when
the \emph{typical} value of $|h_x|$ is large enough. Recall that $p_v^\si$ denotes the probability that $\sigma_{v}=1$ under boundary conditions $\si$.
\begin{definition}
  Let $G=(V,E)$ be a graph, and let $v\in V$ be a vertex.  We say that
  \vocab{strong spatial mixing (SSM) with rate $\alpha(\cdot)$ and
    min-distance $\ell_{0}$} holds for $v$
  if for any $\La\sub V$ and any two configurations $\si_{\La}, \tau_{\La}\in\{\pm 1 \}^{\Lambda}$,
\begin{equation*}
|p_v^{\si_{\La}} - p_v^{\tau_\La}| \le \alpha (d(v,\La'))
\end{equation*}
whenever $d(v,\La')\ge \ell_0$, 
where $\La'\subeq \La$ is the subset on which $\si_{\La}$ and $\tau_\La$ differ.
\end{definition}

The standard definition of \vocab{strong spatial mixing with
    rate $\alpha$} corresponds to taking $\ell_{0}=0$. Taking $\ell_{0}$
  non-zero is a weaker condition. 
The preceding definition is partly inspired by Camia, Jiang, and Newman~\cite{camia2018note}, who obtained a certain \emph{non-uniform} spatial mixing result on $\Z^d$. 
For algorithmic purposes a uniform spatial mixing result is necessary, and we establish such a result in  Lemma~\ref{l:tree-mix} using ideas similar to those of~\cite{camia2018note}. The cost of uniformity is that we obtain SSM with min-distance $\ell_0$ of order $\log n$. Section~\ref{s:ssmsawtree} extends this SSM result to SAW trees, and we prove Theorem~\ref{thm:main} in Section~\ref{s:proofmain}.

For completeness, we show how our generalization of the technical result in~\cite{camia2018note} leads to a generalization of~\cite[Theorem~6]{camia2018note} in Section~\ref{sec:NUSM}. This section does not play a role in our algorithmic results.

\subsection{Disagreement percolation and spatial mixing}

Lemma~\ref{l:bound-by-perc} below relates the distance between marginal distributions to
the probability of disagreement percolation from the boundary to the region of
interest. First, we define the relevant notion of a percolation process.
\begin{definition}
  Let $G=(V\cup \partial V,E)$ be a finite graph.
  Define (inhomogeneous, independent)
  \vocab{site percolation} with probabilities $p_x\in [0,1]$, for each $x\in V$, as the following
  process. Let $T\in \{0,1\}^{V\cup \pl V}$, where
  $(T_x)_{x\in \pl V}$ are given boundary conditions on $\partial V$ and
  $(T_x)_{x\in V}$ are independent Bernoulli random variables with
  $\P(T_x=1)=p_x$.  We denote the
  law of $(T_x)_{x\in V}$ by $P_p$. 
\end{definition}
In the preceding definition the boundary condition is implicit in the notation $P_p$; we will explicitly highlight the boundary condition in what follows. As is standard in percolation, for disjoint $A,B\sub V$, we write $A\lra B$ if
  there exists a path $v_0,v_1,\ldots, v_d$ with $v_0\in A$ and
  $v_d\in B$, such that $T_{v_i}=1$ for each $0\le i\le d$.
  
\begin{lemma}[{cf. \cite[Lemma 5]{camia2018note}}]\label{l:bound-by-perc}
  Given an Ising model on a connected 
  graph $G=(V\cup\partial V,E)$, let $A\sub V$, and let
  $\eta, \xi$ be two boundary conditions on $\partial V$.  Let $P_p$ be the law of a site percolation $T$ with boundary condition 
  $T_x=\one[\eta_x \ne \xi_x]$ for $x\in \pl V$, and
  $p_x = M(\deg(x), h_x, \be)$ for all other vertices.  Then
\begin{equation*}
d_{TV}(p^{\eta}_{\be, h}(\si_A\in \cdot ),p^{\xi}_{\be, h}(\si_A\in \cdot))
\le  P_{p}(\pl V \lra A).
\end{equation*}
\end{lemma}
\begin{proof}
Order the vertices of $V=\{x_1,x_2,\ldots\}$ in such a way that $x$ precedes $y$ in the ordering if $d(x,\pl V)<d(y,\pl V)$. We couple draws $\si^{(1)}$, $\si^{(2)}$, $S_x$ by drawing $\si^{(1)}_x,\si^{(2)}_x, S_x$ sequentially according to an exploration process. $S$ will be a site percolation process with boundary condition given by
\begin{align*}
S_x &= \one [\eta_x\ne \xi_x]\text{ when }x\in \pl V.
\end{align*}
For $t\in \N_0$, let $W_t$ denote the set of sites explored up to and including time $t$, and let $V_t:=\{x\in W_t:S_x=1\}$. We now inductively define the explored set.
\begin{itemize}
\item
Let $W_0:=\pl V$. 
\item
For each $t\ge 0$, reveal the first unexplored site $x$ (according to
the chosen ordering) that is adjacent to $V_t$. Note that this eventually
exhausts the graph since $G$ is connected. We set the values
of $\si_x^{(i)}$ to have the correct marginal distributions, and to be
a maximal coupling. More precisely, let $U_x$ be a independent uniform
random variable in $[0,1]$, let $\nu^{(1)}=\eta$, $\nu^{(2)}=\xi$, and let
\begin{align*}
\si_x^{(i)} &= \begin{cases}
1, & U_x\le p^{\nu^{(i)}}_{\be, h}(\si_x=1|\si_{W_t} = \si_{W_t}^{(i)})\\
-1, &\text{otherwise,}
\end{cases}\\
S_x &= \one[\si_x^{(1)}\ne \si_x^{(2)}].
\end{align*}
Then let $W_{t+1}:=W_t\cup \{x\}$. 
\end{itemize}
Note that conditioned on $\si_{W_t}^{(i)},i=1,2$,
we have $\si_x^{(1)}\ne \si_x^{(2)}$ with probability at most
$M(\deg(x), h_x, \be)$, and as a result the site percolation process $T$ with the same boundary condition as $S$ stochastically
  dominates $S$. 
With this coupling, $\si_A^{(1)}\ne \si_A^{(2)}$ only if $\pl V \lra A$ in $S$; by stochastic domination this is at most the probability that $\pl V \lra A$ in $T$.
\end{proof}

We now use Lemma~\ref{l:bound-by-perc} to show  
how an assumption that the external field is
typically large results in a strong spatial mixing property on trees. We quantify typically large by requiring the following condition on the external field distribution $h$ (for a parameter $h_0$ to be specified):
\begin{equation}
    \label{e:lf}
    \P\pa{|h|< h_0}\le 
    \rc{16 \De^2}.
\end{equation}
\begin{lemma}
  \label{l:tree-mix}
  Let $G$ be a tree with max degree $\De$ and root vertex $v$. 
  Let $h_0$ be such that
  $M(\Delta, h_0,\be) < \De^{-2}$.  
  Suppose $h_x$ are such that
  along each path from any $w$ to $v$, the $h_x$ are independent and satisfy~\eqref{e:lf}. 
Then with probability at least $1-\delta$ over the $h$, there is a $c_1>0$ such that strong spatial mixing holds with rate 
\begin{equation*}
\alpha (\ell) = e^{-c_1 \ell} 
\end{equation*}
for $v$ and for min-distance $\ell_0:=\log_2 \pf{1}{2\delta}$. 
\end{lemma}

\begin{remark}
  We can take $c_1 = -\rc2 \log(M(\De, h_0,\be)\De^2)$.
  Examining the proof shows that the right-hand side of~\eqref{e:lf} can be improved to $O\prc{\De^{1+\ep}}$ for any $\ep>0$ at the cost of a tighter bound on $M(\De, h_0,\be)$, by replacing $\fc{\ell}2$ in~\eqref{e:CH} by $(1-O(\ep)) \ell$. 
\end{remark}
\begin{proof}
First, we fix $\ell\ge \ell_0$ and consider the case when $\si_\La$ and $\tau_\La$ disagree at exactly 1 vertex $w$ at distance $\ell$ from $v$. There is a unique path $\ga_{wv}$ joining $w$ and $v$. If $\ga_{wv}$ contains another vertex $w'\in\La$ besides $w$, then the probability of the spin at $u$ is conditionally independent of the spin at $w$ given the spin at $w'$, so $p^{\si_\La}_u=p^{\tau_\La}_u$. Otherwise, we apply Lemma~\ref{l:bound-by-perc}: letting $p_x=M(\deg(x), h_x,\be)$ for $x\in V\bs \La$, we have 
\begin{align*}
|p_v^{\si_\La} - p_v^{\tau_\La}| &\le
    P_{p}(w\lra v)
    \le \prod_{u\in \ga_{wv}} M(\Delta, h_u,\be),
\end{align*}
since percolation occurs only if all sites on the path $\ga_{wv}$ have value 1; the second inequality is by Lemma~\ref{lem:Mmon}.
In the general case, by changing the vertices at distance $\ell$ one at a time, we get that 
\begin{align*}
|p_v^{\si_\La} - p_v^{\tau_\La}| &\le \sum_{w:d(v,w)=\ell} \prod_{u\in \ga_{wv}} M(\De,h_u,\be).
\end{align*}
To estimate this we next establish that most $h_u$ are large.  Let $p = \mathbb{P}(|h|<h_0)$. Using the upper bound $p\le \rc{16\De^2}$ from \eqref{e:lf} and
the Chernoff-Hoeffding bound we obtain that, with high probability, most of the $h_u$'s along a path are large. 
\begin{align}
\label{e:CH}
\P\pa{\abb{\{u\in \ga_{wv}:|h_u|\ge h_0\}}\ge \fc \ell2} &\ge
1-e^{-\pa{\rc 2 \log \pf{1/2}p + \rc 2 \log\pf{1/2}{1-p}}\ell}\\
\nonumber
&\ge 1-e^{\pa{\log 2 - \rc2\log \prc{p}}\ell} \\
\nonumber
&\ge 1-2^\ell p^{\rc 2\ell }
\ge 1-\prc{\De}^\ell 2^{-\ell}.
\end{align}

Under this event and by the hypothesis that $M(\Delta,h_0,\beta)<\Delta^{-2}$, we have that there exists $c_1>0$ such that 
\begin{equation*}
\prod_{u\in \ga_{wv}} M(\Delta,h_u,\be) \le \pf{e^{-2c_1}}{\De^2}^{\ell/2}
= \fc{e^{-\ell c_1}}{\De^\ell}.
\end{equation*}
In particular, $c_1=-\rc2 \log(M(\De, h_0,\be)\De^2)$ works. 
Hence, doing a union bound over all paths, we get that with probability at least 
$1- 2^{-\ell}$, 
\begin{align*}
|p_v^{\si_\La} - p_v^{\tau_\La}| &\le e^{-\ell c_1}.
\end{align*}
for any $\si_{\La},\tau_{\La}$ disagreeing in $\La'$, where $d(\La,\La')=\ell$. Now taking a union bound over $\ell\ge \ell_0$, this holds for all $\ell\ge \ell_0$ with probability 
\begin{align*}
1-\sum_{\ell=\ell_0}^{\iy} 2^{-\ell}
\ge 1-2 \cdot 2^{-\ell_0} \ge
1-\delta.
\end{align*}
where the last inequality uses the definition of $\ell_0$.
\end{proof}

\subsection{Strong spatial mixing on the SAW tree}
\label{s:ssmsawtree}

Our next corollary will have nearly the same conclusion as Lemma~\ref{l:tree-mix}, 
but it concerns a specific tree of self-avoiding walks
(SAW tree).  
To prepare for this, we recall the construction of the SAW tree for Ising models,
see~\cite[Appendix~A]{liu2019fisher} or~\cite{zhang2011approximating},
which follows Weitz's original construction for the
hard-core model~\cite{weitz2006counting}. As this construction is
well-known and clear expositions exist in the literature, we will be
somewhat brief. 

A \vocab{self-avoiding walk of length $k$}, $(v_{i})_{i=0}^{k}$, is a
sequence of adjacent vertices, each $v_{i}$ distinct. The set of
self-avoiding walks started at a fixed vertex $v$ has a natural rooted
tree structure: the root is the length $0$ walk consisting of
$v_{0}=v$ alone, and the children of length $k$ self-avoiding walk
$(v_{i})_{i=0}^{k}$ are the length $k+1$ self-avoiding extensions
$(v_{i})_{i=0}^{k+1}$ with $v_{k+1}$ adjacent to $v_{k}$. Call this
tree $\hat T_v$. The \vocab{self-avoiding walk tree $T_{v}$ rooted at
  $v$} is obtained from $\hat T_v$ as follows. To each vertex
$(v_{i})_{i=0}^{k}$ in $\hat T_v$ append additional leaf vertices
$w_{1},\dots, w_{j}$, one for each
$v_{k+1}\in V, v_{k+1}\neq v_{k-1}$, such that $(v_{i})_{i=0}^{k+1}$
is \emph{not} self-avoiding, i.e., $v_{k+1}$ completes a cycle of
length at least three. The tree $T_{v}$ is finite, and one obtains an
Ising model on $T_{v}$ by taking the external field at
$(v_{i})_{i=0}^{k}$ to be $h_{v_{k}}$. 

The key result is then that if one defines a boundary condition $\tau$
on the leaves of $T_{v}$ correctly (i.e., $\tau_{w}\in \{-1,+1\}$ for
each leaf $w$, see Appendix~\ref{a:saw} 
or \cite{liu2019fisher,zhang2011approximating} for details of the construction), 
then the distribution of
$\sigma_{v}$ on $T_{v}$ with boundary condition $\tau$ is identical to
the distribution of $\sigma_{v}$ on $G$:
\begin{lemma}
  \label{lem:Wlemma}
  There is a choice of spins $\tau_{w}$ for the leaves $w$ of $T_{v}$
  such that the marginal distribution of the spin at the root is
  precisely the same as the marginal distribution of $\sigma_{v}$ on
  the graph $G$.
\end{lemma}
The exact way in which the boundary condition $\tau$ is determined
will not play a role in what follows, so we will not discuss
this in the main text. 
The preceding construction generalizes to the
situation in which there is a boundary condition $\xi$ for the Ising
model on $G$: in this case the corresponding spins for the Ising model
on $T_{v}$ are fixed to agree with $\xi$. Lemma~\ref{lem:Wlemma} holds
in this more general setting as well.

\begin{corollary}
  \label{cor:ssmdSAW}
  Let $G$ be a graph with max degree $\De$ on $n$ vertices.   Suppose that the distribution of $h_x$ satisfies \eqref{e:lf} for $h_0$ such that $M(\Delta,h_0,\beta)<\De^{-2}$.
  Then there exists a constant $c_1>0$ so that with probability $ 1-o(1)$ over the realization of the external fields, for every vertex $v\in V$ the SAW tree $T_v$ at
  $v$ satisfies strong spatial mixing with rate
  \begin{equation*}
    \alpha (\ell) = e^{-c_1\ell}
  \end{equation*}
  for $v$ and for min-distance $\ell_0:=\fc{\log n}{c_1}$.
\end{corollary}

\begin{proof}
  Note that in the SAW tree rooted at $v$ the value of
  $h_x$ is repeated at some vertices, but there are no repetitions
  along any path to   the root $v$, because this path corresponds to a self-avoiding
  walk. Therefore 
  Lemma~\ref{l:tree-mix} applies to the SAW tree rooted at
  $v$. To obtain the corollary, apply the result of
  Lemma~\ref{l:tree-mix} with $\delta$ taken to be
  $\fc{\delta}n$ for each vertex $v\in
  V(G)$. The result follows by a union bound over all vertices of $G$, choosing $\delta=o(1)$, and decreasing $c_1>0$ if necessary.
\end{proof}

\subsection{Proof of Theorem~\ref{thm:main}}
\label{s:proofmain}

We use Weitz's approach to approximate counting~\cite{weitz2006counting}.    In brief, to approximate the partition function of a graph in $ \mathcal G_\De$ with arbitrary boundary conditions it suffices to approximate the marginal of any vertex of any graph in $\mathcal G_\De$ with arbitrary boundary conditions by writing the partition function as a telescoping product; see Appendix~\ref{a:saw} for the calculation.  Similarly, given the ability to approximate marginals, one can sample by setting one spin at a time according to its marginal and updating the boundary conditions. Both the algorithms for counting and for sampling are written out in Appendix~\ref{a:alg}.

Recall,
see~\cite{weitz2006counting}
or~\cite[Theorem~2.8]{sinclair2014approximation} that Weitz proved
that for any two-state spin system, strong spatial mixing (SSM) on the
$\Delta$-regular tree implies the existence of an FPTAS on all graphs
of degree at most $\Delta$. In the following, we briefly recall this
algorithm and its analysis. We will see that weaker notion of SSM at
distance $\ell_{0}$ is sufficient to carry out the analysis.

By the SAW tree construction discussed in 
  Section~\ref{s:ssmsawtree}, to compute the marginal distribution of
  the spin at a fixed vertex $v$ it would suffice to compute the
  marginal distribution of the spin at the root of $T_{v}$. Given the
  tree structure, this is a recursive computation. The running time,
  however, is exponential in the depth of the recursion, which could
  be as large as $n$. Weitz observed that one can truncate the tree
  $T_{v}$ at logarithmic depth if  correlations decay exponentially fast in the depth of the tree, as the
  analysis of the recursive computation will not be sensitive to the
  value of the spins at large distances. The running time of the
  recursion on this truncated tree is linear in the size of the tree, which is polynomial in $n$.

To make this precise for the hard-core model, Weitz proved that when $\lambda < \lambda_c(\Delta)$, SSM with rate $\alpha(t) = e^{- \Omega(t)}$ holds on the SAW
tree of a graph of maximum degree $\Delta$. Consequently, to
obtain an $\epsilon/n$-approximate evaluation of the marginal of the root of the SAW tree, one can truncate the SAW tree at
depth $\ell = O(\log(n/\epsilon))$
(see~\cite[Section~5]{weitz2006counting}). The running time
of this algorithm is polynomial in $n$.

\begin{proof}[Proof of Theorem~\ref{thm:main}]
Given the discussion above and Lemma~\ref{lem:Wlemma}, it is clear that Corollary~\ref{cor:ssmdSAW} suffices for verifying
that validity of the algorithm in the setting of the RFIM: for a fixed $0<\delta<1$, we have
SSM at distance $\ell_{0}=\frac{\log n}{c_1}$. 
Hence by taking $\ell'=\max\{\ell,\ell_{0}\}$ and truncating at depth $\ell'$ we obtain the desired polynomial-time
algorithm.
\end{proof}

Lastly, we check that the event on which 
the algorithm
correctly outputs the desired approximation can be
identified in polynomial time as was described in Remark~\ref{rem:check}. 
\begin{proposition}
  \label{prop:check}
  Fix $\epsilon>0$. Given $h$, there is a polynomial time algorithm in
  $1/\epsilon$ and $n$ that determines if the
  output of the algorithm from Theorem~\ref{thm:main} is an
  $\epsilon$-approximation to the partition function of the Ising
  model with external fields $h$.
\end{proposition}
\begin{proof}
  For each $v\in V$, it takes polynomial time to construct the SAW
  tree $T_{v}$ to depth $\ell'=\max\{\ell,\ell_{0}\}$, where $\ell$ is
  the constant from the proof of Theorem~\ref{thm:main} above. Constructing the
  SAW tree for each $v$ to this depth thus takes polynomial time as
  well. Each tree has polynomially many leaves, and for each leaf to
  check that the path from root to leaf has at least half of its
  vertices with external field at least $h_{0}$ in magnitude takes
  linear time. 
\end{proof}

\section{Random field Ising model on random graphs}
\label{sec:extensions}

This section establishes Theorem~\ref{thm:main2}.
The arguments are similar to those that established
Theorem~\ref{thm:main}, and we focus our exposition on the new
aspects.

A graph drawn from $\mathcal G(n,p)$ (an Erd\H{o}s-R\'enyi random graph)
is defined as a graph on $n$ vertices $\{v_1,\ldots, v_n\}$ where each pair of vertices $\{v_i,v_j\}$, $i\ne j$ forms an edge independently with probability $p$. We will take $p=\De/n$, so that the average degree of a vertex is $p(n-1)\approx \De$. However, the maximum degree of $\mathcal G(n,\De/n)$ is $\Te\pf{\log n}{\log\log n}$ with high probability.

The key observation that allows us to apply our results in this setting is that an Erd\H{o}s-R\'enyi random graph has bounded connective constant with high probability~\cite{sinclair2013spatial,sinclair2017spatial}. 
\begin{lemma}[{\cite[in Proof of Theorem~1.2]{sinclair2013spatial}}]
\label{l:cc}
Let $\De >1$. 
Suppose $G$ is distributed as $\mathcal{G}(n,\Delta/n)$, and let $\gamma,\nu>0$. 
With probability at least $1-n^{-\nu}$, for all $\ell \ge \fc{\nu+2}{\log(1+\gamma/2)}\log n$ and all vertices $v$, the SAW tree $T_v$ satisfies 
\begin{equation*}
|N(v,\ell)| \le [\De(1+\gamma/2)]^\ell, \qquad 
\sum_{d=1}^\ell |N(v,d)| \le \fc{\De}{\De-1}[\De(1+\gamma/2)]^\ell.
\end{equation*}
\end{lemma}

The following lemma is the analogue of Corollary~\ref{cor:ssmdSAW}. The key additional idea is to classify a vertex as bad if its degree is large, and that large degrees occur with low probability in $\mathcal{G}(n,p)$. Note that degrees of vertices are not independent, so we instead use the Chernoff–Hoeffding inequality on the edge count.
\begin{lemma}\label{l:ER}
 Let $\De>1$ and $G\sim \mathcal G(n,\Delta/n)$ be a random graph on $V$. There are constants $c_3, c_4, c_5$ such that the following hold. Fix a vertex $v\in V$. Let $c_1,c_2$ be any large enough constants. Let $h_0$ be such that $M(e^{c_2c_3}\Delta, h_0, \beta) \le \fc{e^{-c_1}}{c_4\Delta^2}$ and $\P(h_u<h_0)\le \pf{1}{2\De}^{c_2c_5}$.
 Then with probability $\ge 1-\de$ over $G$ and the realization of the external field $h$, the following hold.
 \begin{enumerate}
     \item The SAW tree $T_v$ at $v$ satisfies strong spatial mixing with rate $\al(\ell) = e^{-c_1\ell/2}$ for $v$ and for min-distance $\ell_0=\fc{\log\pf{2}{\de}}{c_2}$. 
     \item For all $\ell \ge \ell_0$, on the SAW tree $T_v$, $\sum_{d=1}^{\ell} |N(v,d)|\le (c_4\De)^{\ell/2}$. 
 \end{enumerate}
This holds for the SAW trees at all vertices with the same rate and min-distance $\ell_0=\fc{\log\pf{2n}{\de}}{c_2}$. 
\end{lemma}
\begin{proof}
Call a path $v_0,\ldots, v_\ell$ on the SAW tree starting at $v_0=v$ \emph{bad} if one of the following holds:
\begin{enumerate}
\item
At least $\fc 14\ell$ of the vertices $v_0,\ldots, v_{\ell-1}$ have degree  greater than $c_3\Delta$.
\item
At least $\fc 14\ell$ of the values $|h_{v_i}|$ satisfy $|h_{v_i}|> h_0$.
\end{enumerate}
The first step in the proof is to rule out the existence of bad paths with high probability by a union bound argument. To this end, we first bound the probability (over the randomness in $G$ and $h$) that a fixed sequence $v_0,\ldots, v_\ell$ of vertices is a bad path. Since each edge is included in $G$ with probability $\fc{\De}n$, the probability of this fixed path 
being a path in the SAW tree is $\pf{\De}{n}^\ell$. 

Next we bound the probability of event (i). Given a fixed subset $S$ of $\{v_0,\ldots, v_{\ell-1}\}$ with $\ce{\ell/4}$ vertices, we bound the probability that all its vertices have degree greater than $e^{c_2c_3}\Delta$. For this to happen, there must be at least $\fc{e^{c_2c_3}\Delta}2\cdot \fc{\ell}{4}$ edges in $S\times (V\bs S)$, which has cardinality $O(\ell n)$. By the Chernoff-Hoeffding bound, the probability of this is 
$e^{-\Omega(\Delta \ell c_2c_3)}$
for large enough $c_3$. By a union bound over appropriate subsets of $S$ (less than $2^\ell$ in number), the probability of (i) is still $e^{-\Omega(\Delta \ell c_2c_3)}$.

The probability of event (ii) is bounded exactly as in Lemma~\ref{l:tree-mix}: letting $p=\pf{1}{2\De}^{c_2c_5}$, it is bounded by 
$p^{O(\ell)} \le (2\De)^{-\Om(c_2c_5\ell)}$. Thus, the probability of a fixed sequence $v_0,\ldots, v_\ell$  being a bad path is 
\begin{align*}
\pf{\De}{n}^\ell \cdot \pa{e^{-\Omega(\Delta \ell c_2c_3)} +  (2\De)^{-\Om(c_2c_5\ell)}}\le \rc2
\prc{n}^{\ell}\cdot e^{-c_2\ell }.
\end{align*}
for large enough $c_3,c_5$.

By a union bound over possible paths of length $\ell$, of which there are most $n^\ell$, the probability that a bad path of length $\ell\ge \ell_0$ exists is at most $e^{-c_2\ell_0 }$. This is at 
most $\fc{\de}2$ when $\ell_0=\log\pf2{\de}/c_2$. 

Next, observe that we can choose $c_4$ sufficiently large so that by Lemma~\ref{l:cc}, the probability that $\sum_{d=1}^\ell |N(v,d)|\le (c_4 \De)^{\ell/2}$ for all $\ell\ge \ell_0$ 
is at least $1-\fc{\de}2$. (We apply Lemma~\ref{l:cc} directly when $\De\ge 2$; otherwise to avoid the $\rc{\De-1}$ factor, we note that $\sum_{d=1}^\ell |N(v,\ell)|$ is stochastically dominated by its value when $\De=2$.) 
Let $\mathcal{E}$ be the event that there are no bad paths of length $\ell\ge \ell_0$, and 
$\sum_{d=1}^\ell |N(v,d)| \le (c_4 \De)^{\ell/2}$ 
for each $\ell\ge \ell_0$. By a union bound, $\mathcal{E}$ has probability at least $1-\de$.

If a path from $v$ to $w$ is not bad, then letting $\ell = d(v,w)$, at least $\fc 12 \ell$ of the vertices $u$ on the path satisfy $\deg(u)\le e^{c_2c_3}\Delta$ and $|h_u|\le h_0$. For any such $u$,  $M(\deg(u), h_0, \beta)\le \fc{e^{-c_1}}{c_4\Delta^2}$ by Lemmas~\ref{l:M} and~\ref{lem:Mmon}.
As in the proof of of Lemma~\ref{l:tree-mix}, 
let $\ga_{wv}$ denote the unique path from $w$ to $v$.
Then, we have that if $\si_\La$, $\tau_\La$ are boundary conditions differing only at $w$, then 
\begin{align*}
\P(\si_v\ne \tau_v|\si_\La, \tau_\La) &= 
\prod_{u\in \ga_{wv}}M(\deg(u), h_{v_k},\be)\le \pf{e^{-c_1}}{c_4\Delta^2}^{\ell/2} = \fc{e^{-\ell c_1/2}}{c_4^{\ell/2} \De^\ell}. 
\end{align*}
Fix $\ell \ge \ell_0$. On the event $\mathcal{E}$, 
by changing the vertices at distance $\ell$ one at a time, we have
\begin{equation*}
|p_v^{\si_\La} - p_v^{\tau_\La}| \le \sum_{w:d(v,w)=\ell} \sum_{u\in \ga_{wv}} M(\deg(u), h_u,\be)
\le (c_4 \De)^{\ell/2} \cdot \fc{e^{-\ell c_1/2}}{c_4^{\ell/2} \De^\ell} \le e^{-\ell c_1/2},
\end{equation*}
and we obtain the same conclusion for all $v$ on the event that no bad paths exist starting from any $v$. 
Finally, by replacing $\de$ by $\fc{\de}n$ and union-bounding over all vertices, there are no bad paths in the SAW tree at $v$ for each vertex $v$. 
\end{proof}

\begin{proof}[Proof of Theorem~\ref{thm:main2}]
Theorem~\ref{thm:main2} follows from Lemma~\ref{l:ER} after noting the high-probability bound on the neighborhood of a vertex $v$ given by Lemma~\ref{l:cc}.  The recursive computation of the marginals on the SAW tree still runs in polynomial time  because the size of the $\ell'$-neighborhood of $v$ in the SAW tree $T_v$ is  $\sum_{d=1}^{\ell'}|N(v,d)| \le (c_4\De)^{\ell'/2}$, which is polynomial in all parameters.
\end{proof}

\section{Non-uniform spatial mixing on infinite graphs}
\label{sec:NUSM}\label{s:non-unif}
The following non-uniform spatial mixing result generalizes~\cite[Theorem~6]{camia2018note}. In this section we work in the context of infinite graphs; we always fix boundary conditions on sets $B$ such that $V\setminus B$ is finite.

\begin{theorem}[{cf.\ \cite[Theorem~6]{camia2018note}}]\label{t:cjn}
Consider the RFIM with IID Gaussian external fields. 
Let $c_2>0$. There exists $c_1(\De, \be, c_2)$ such that for $\Var(h_x)\ge c_1(\De, \be,c_2)$, for any $A$ and $B$ such that $V\bs B$ is finite, for almost all realizations $h$,
\begin{align*}
\sup_{\eta, \xi \in \{\pm 1\}^{B}}
d_{TV}(p^{\tau\wedge \eta}_{\be, h}(\si_A\in \cdot ),p^{\tau \wedge \xi}_{\be, h}(\si_A\in \cdot))
&\le 
\sum_{x\in \pl A, y\in \pl B} c_3(x,h) e^{-c_2d(x,y)}. 
\end{align*}
where 
\begin{align*}
(\tau\wedge\eta)(x)
&= \begin{cases}
\tau(x), & x\in V\setminus B\\
\eta(x), & \text{otherwise,}
\end{cases}
\end{align*}
and similarly for $\xi$.
Here, for a set $A\sub V$, $\pl A \subset V\setminus A$ denotes the subset of $A$ whose neighbors are not all contained in $A$.
\end{theorem}

Recall the inhomogenous site percolation processes $P_p$ with $p_x = M(\deg(x),h_x,\beta)$ introduced in Lemma~\ref{l:bound-by-perc}. 
\begin{lemma}[{cf. \cite[Lemma 6]{camia2018note}}]
\label{l:6}
Consider the measure $P_p$ averaged over the randomness in $h$,  $\ol P_p(\cdot) = \int_{\R^V}P_p(\cdot) \P(dh)$ where $\P$ is the law of $h$.
Let $c_2>\log 2$. There exists $h_0=h_0(\De, \be,c_2)$ so that when $\P(|h_u|<h_0)\le p:=\fc{e^{-2c_2}}{4\De^2}$, 
\begin{align*}
\ol P_p(x\lra y)
&\le
 4e^{-c_2d(x,y)}.
\end{align*}
\end{lemma}

\begin{proof}
Choose $h_0 = |\be| \De + \log(\De) + c_2$, so by Lemma~\ref{l:M}, $M(\De, h_0, \be) < \fc{e^{-2c_2}}{\De^2}$.  
Consider a path $\ga$ from $x$ to $y$ of length $\ell: = d(x,y)$. Then 
\begin{align}\label{e:avg-perc}
\ol P_p(\{\forall u\in \ga, S_u = 1 \})
&= \prod_{u\in \ga} M(\deg(u), h_u,\be).
\end{align}
As in~\eqref{e:CH}, by the Chernoff bound, with high probablity, most of the $h_u$'s for $u\in \ga$ are large:
\begin{align*}
\P\pa{\abb{\{u\in \ga:|h_u|\ge h_0\}}\ge \fc \ell2} &\ge 
1-e^{-\pa{\rc 2 \log \pf{1/2}p + \rc 2 \log\pf{1/2}{1-p}}\ell}\\
&\ge 1-2p^{\rc 2\ell} 
\ge 1-\pf{e^{-c_2}}{\De}^\ell 
\end{align*}
Under this event, 
\begin{align*}
\prod_{u\in \ga} M(\De, h_u, \be) &\le \pf{e^{-2c_2}}{\De^2}^{\ell/2} = \fc{e^{-c_2\ell}}{\De^\ell}
\end{align*}
Hence, by breaking up~\eqref{e:avg-perc} into two bad events,
\begin{align*}
\ol P_p(\{\forall u\in \ga, S_u = 1 \})
&=\P\pa{\abb{\{u\in \ga:|h_u|\ge h_0\}}< \fc \ell2}\\
&\qquad + 
\ol P_p\pa{\{\forall u\in \ga, S_u = 1 \} \Big| \abb{\{u\in \ga:|h_u|\ge h_0\}}\ge \fc \ell2}\\
&
\le \fc{e^{-c_2\ell}}{\De^\ell} +  \fc{e^{-c_2\ell}}{\De^\ell}  = \fc{2e^{-c_2\ell}}{\De^\ell}.
\end{align*}
There are at most $\De^j$ paths of length $j$, so taking a union bound over all paths gives
\begin{align*}
\ol P_p(x\lra y)
&\le \sum_{j=\ell}^\iy \De^j \fc{2e^{-c_2j}}{\De^j} \le 4e^{-c_2\ell}.\qedhere
\end{align*}
\end{proof}

Given Lemma~\ref{l:6}, the proof of Theorem~\ref{t:cjn} is exactly the same as in~\cite{camia2018note}, after noting that the neighborhood of a vertex grows at most exponentially.

\begin{proof}[Proof of Theorem~\ref{t:cjn}]
By Lemma~\ref{l:6} applied to $c_2\mapsfrom 2c_2+\log(\De)$, 
\begin{align*}
\sum_{y\in V} e^{c_2d(x,y)} \ol P_p(x\lra y) & \le 
\sum_{y\in V} e^{c_2d(x,y)} 4 e^{-2c_2d(x,y)} \De^{-d(x,y)} <\iy,
\end{align*}
where we use the fact that the number of vertices at distance $\ell$ from $x$ is at most $\Delta^\ell$. 
Expanding $\ol P_p$ as an integral over $h$ and using the Fubini-Tonelli theorem,
\begin{align*}
\int_{\R^V} \sum_{y\in V} e^{c_2d(x,y)} P_p(x\lra y) \P(dh)&<\iy,
\end{align*}
which implies 
\begin{align*}
\sum_{y\in V}e^{c_2d(x,y)} P_p(x\lra y) &<\iy
\end{align*}
for almost all $h$, and 
\begin{align*}
e^{c_2d(x,y)} P_p(x\lra y) &<c_3(x,h) \text{ for all }y\in V
\end{align*}
for almost all $h$. 
Using Lemma~\ref{l:bound-by-perc}, we get 
\begin{align*}
\sup_{\eta, \xi \in \{\pm 1\}^{B}}
d_{TV}(p^{\tau\wedge \eta}_{\be, h}(\si_A\in \cdot ),p^{\tau\wedge \xi}_{\be, h}(\si_A\in \cdot))
&\le P_p (A\lra B)\\
&\le \sum_{x\in \pl A, y\in \pl B} P_p (x\lra y)\\
&\le \sum_{x\in \pl A, y\in \pl B} c_3(x,h)e^{-c_2d(x,y)}.\qedhere
\end{align*}
\end{proof}

\section*{Acknowledgements}

This work was undertaken as part of the Phase Transitions and Algorithms working group in the SAMSI Spring 2021 semester program on Combinatorial Probability.  We thank the semester organizers for bringing us together.

\bibliographystyle{plain}
\bibliography{bib}

\appendix

\section{SAW tree and recursion}
\label{a:saw}

In this section, we first describe how the Ising model on a graph $G$ translates to an Ising model on the SAW tree $T_v$ defined in Section~\ref{s:ssmsawtree}, and then show how to compute the marginal probabilities for $\si_v$ in $T_v$.

\paragraph{Ising model on the SAW tree.}
Given an Ising model on a graph $G=(V,E)$ with inverse temperature $\be$ and external fields $h$, we obtain an
Ising model on $T_{v}$ by taking the same inverse temperature $\be$ and taking the external field at
$(v_{i})_{i=0}^{k}$ to be $h_{v_{k}}$. If we are given boundary conditions $\tau$ on a set of vertices $\pl V$, then we take the boundary condition at 
$(v_{i})_{i=0}^{k}$ to be $\tau_{v_k}$ whenever $v_k\in \pl V$. 

Fix a lexicographic order on the vertices $V$; this induces an order on the edges incident to any fixed vertex.
For the vertices $\ga$ in $T_v\setminus \hat T_v$ (those representing paths with a cycle), we instead assign them the following boundary condition:
\begin{align*}
\tau_\ga &= \begin{cases}
1, &\text{if the edge closing the cycle is larger than the edge starting the cycle in }\ga\\
-1, &\text{otherwise.}
\end{cases}
\end{align*}

\paragraph{Marginal probabilities on the SAW tree.}
For simplicity of notation, let $p=p_{T_v, \be, h}$ denote the Ising model on the $T_v$.
For convenience, we work with the occupation ratio
\[
R_v := \frac{p(\sigma_v=-1)}{p(\sigma_v=+1)}.
\]
Recalling the definition of Gibbs measure, we have 
\[
R_v =e^{-2h_v}\frac{ p'(\sigma_v=-1)}{p'(\sigma_v = +1)} =e^{-2h_v} \prod_{i=1}^{\deg(v)} \frac{p_{u_i}'(\sigma_v =-1)}{p_{u_i}'(\sigma_v =+1)}, 
\]
where $p'$ is the Gibbs measure after removing $h_v$ (the external field at vertex $v$), and $u_i$ are the child vertices of $v$. 
The measure $p'_{u_i}$ is the Gibbs measure defined on the subtree by removing all other subtrees except the one rooted at 
vertex $u_i$. Note that this still includes the root vertex $v$. 
We define $p''_{u_i}$ as the Gibbs measure on the subtree rooted at $u_i$, excluding $v$. Then 
\[
\frac{p_{u_i}'(\sigma_v =-1)}{p'_{u_i}(\sigma=+1)} = \frac{e^{\beta}p''_{u_i}(\sigma_{u_i}=-1)+e^{-\beta}p''_{u_i}(\sigma_{u_i}=+1)}{e^{-\beta}p''_{u_i}(\sigma_{u_i}=-1)+e^{\beta}p''_{u_i}(\sigma_{u_i}=+1)} = \frac{e^{2\beta}R_{u_i}+1}{R_{u_i}+e^{2\beta}}.
\]
Using $R_v = \fc{p_v}{1-p_v}$ and $p_v = \rc{1+\rc{R_v}}$, we can write this in terms of $p_v :=p(\sigma_v = 1) $, 
\begin{align}
\label{e:recursion}
p_v  &=
\rc{1+\rc{e^{-2h_v} \prodo id \fc{e^{2\be}R_{u_i}+1}{R_{u_i}+e^{2\be}}}} = 
\frac{1}{1+ e^{2h_v}\prod_{i=1}^d \frac{e^{2\beta} 
(1-p_{u_i})
+ p_{u_i} 
}{
(1-p_{u_i})
+e^{2\beta}p_{u_i}}
}.
\end{align}
This equation provides a recursive method for computing marginal probabilities. Given boundary conditions $\tau$ on $\pl V$, we set $p_v=1$ or 0 according to whether $\tau_v=1$ or $\tau_v=-1$, and then work our way up to the root vertex.

\section{Algorithms}

\label{a:alg}

We explicitly write out the algorithms for approximate sampling (Algorithm~\ref{a:sampling}) and computation of $Z_{G,\be, h}$ (Algorithm~\ref{a:counting}). These algorithms work for both max-degree $\De$ graphs in Theorem~\ref{thm:main} and $\mathcal G(n,\De/n)$ graphs in Theorem~\ref{thm:main2}, in the appropriate regime and with the appropriate constants.
For the sampling algorithm (Algorithm~\ref{a:sampling}), we repeat the following: estimate the marginal probabilities for an unfixed vertex, use it to sample the spin for the vertex, and then add that value to the boundary conditions. 
For estimation of $Z_{G,\be, h}$ (Algorithm~\ref{a:counting}), to see that the product $e^{-H_{\be, h}(\si)} \prod_{i=1}^n r_i$ gives the right answer, note that if $p_{v_i}^*$ are the actual probabilities,
\begin{align*}
    r_i^* :&= \begin{cases}
    \rc{p_{v_i}^*}, & \si_i=1\\
    \rc{1-p_{v_i}^*}, & \si_i=-1
    \end{cases}
    \\
    &= p(\si_i' = \si_i | \si_j' = \si_j \text{ for }j<i) = 
    \fc{p(\si_j'=\si_j \text{ for } j\le i-1)}{p(\si_j'=\si_j \text{ for } j\le i)}.
\end{align*}
Then we have a telescoping product
\begin{align*}
    e^{-H_{\be, h}(\si)} \prod_{i=1}^n r_i^*
    &= 
    Z_{G,\be, h} \cdot p(\si) \prod_{i=1}^n \fc{p(\si_j'=\si_j \text{ for } j\le i-1)}{p(\si_j'=\si_j \text{ for } j\le i)} = Z_{G,\be, h}.
\end{align*}
With appropriate choice of $c$, we can ensure that for each $i$, with probability at least $1-\de/n$, that 
$p_{v_i}\in [p_{v_i}^* e^{-\fc{\ep}{2n}}, p_{v_i}^* e^{-\fc{\ep}{2n}}]$ and 
$r_i\in [r_i^* e^{-\ep/n}, r_i^* e^{\ep/n}]$. Then with probability at least $1-\de$, the estimate will be contained in $r_i\in [Z_{G,\be, h} e^{-\ep}, Z_{G,\be, h} e^{\ep}]$.

 \renewcommand{\algorithmicrequire}{\textbf{Input:}}
 \renewcommand{\algorithmicensure}{\textbf{Output:}}

 \begin{algorithm}[h!]
 \caption{Approximate sampling from RFIM} 
 \begin{algorithmic}[1]
 \Require Random field Ising model $(G, \be, h)$, failure probability $\de$, accuracy $\ep$.
 \Ensure Approximate sample
 \medskip  
 \State Order the vertices $v_1,\ldots, v_n$.
 \For{$i = 1 \to n$} 
 	\State Construct the SAW tree at $v_i$, $T_{v_i}$.
 	\State Set boundary conditions $\tau_w'=1$ in 
 	$T_{v_i}$ for all 
 	$w$ such that $d(v_i, w) > 
 	c\max\bc{\log\pf n\ep , \log \pf n\de}$ for an appropriately large constant $c$. 
 	Set $p_w=1$ for these $w$.
 	\Comment{Note that arbitrary boundary conditions can be chosen.}
 	\State Set boundary conditions corresponding to $\si_j, 1\le j<i$ in $T_{v_i}$.
 	\State Use recursion~\eqref{e:recursion} to compute $p_{v_i}$.
 	\State Set 
 	\begin{align*}
 \si_i &= \begin{cases}
 1,&\text{with probability }p_{v_i}\\
 -1,&\text{with probability }1-p_{v_i}.
 \end{cases}
 \end{align*}
 \EndFor
 \State \Return $(\si_1,\ldots, \si_n)$. 
 \end{algorithmic}
 \label{a:sampling}
 \end{algorithm}

 \begin{algorithm}[h!]
 \caption{Approximation of partition function for RFIM} 
 \begin{algorithmic}[1]
 \Require Random field Ising model on $G$, failure probability $\de$, accuracy $\ep$ (where desired multiplicative accuracy is $e^\ep$).
 \Ensure Approximation of partition function $Z_{G,\be,h}$
 \medskip  
 \State Order the vertices $v_1,\ldots, v_n$.
 \For{$i = 1 \to n$} 
 	\State Construct SAW tree at $v_i$,  $T_{v_i}$. 
 	\State Set boundary conditions $\tau_w'=1$ in $T_{v_i}$ for all 
 	$w$ such that $d(v_i, w) > 
 	c\max\bc{\log\pf n\ep , \log \pf n\de}$ for an appropriately large constant $c$. Set $p_w=1$ for these $w$.
 	\Comment{Note that arbitrary boundary conditions can be chosen.}
 	\State Set boundary conditions corresponding to $\si_j, 1\le j<i$ in $T_{v_i}$.
 	\State Use recursion~\eqref{e:recursion} to compute $p_{v_i}$.
 	\If{$p_{v_i}\ge \rc 2$}
 		\State Set $\si_i=1$ and $r_i = \fc{1}{p_{v_i}}$.
 	\Else
 		\State Set $\si_i=-1$ and $r_i = \fc{1}{1-p_{v_i}}$.
 	\EndIf
 	\State Include $\si_i$ as boundary condition in $G'$.
 \EndFor
 \State \Return $e^{-H_{\be, h}(\si)} \prod_{i=1}^n r_i$. 
 \end{algorithmic}
 \label{a:counting}
 \end{algorithm}

\end{document}